\documentclass[conference]{IEEEtran}
\IEEEoverridecommandlockouts
\usepackage{cite}
\usepackage{amsmath,amssymb,amsfonts}
\usepackage{algorithmic}
\usepackage{graphicx}
\usepackage{physics}
\usepackage{tikz}
\usetikzlibrary{arrows.meta,positioning,fit,calc}
\usepackage{textcomp}
\usepackage{bm}
\usepackage{xcolor}
\usepackage{stfloats}
\usepackage[section]{placeins}
\usepackage{amsthm}
 \usepackage{comment}
\theoremstyle{plain}
\newtheorem{theorem}{Theorem}

\theoremstyle{remark}
\newtheorem{remark}{Remark}

\def\BibTeX{{\rm B\kern-.05em{\sc i\kern-.025em b}\kern-.08em
    T\kern-.1667em\lower.7ex\hbox{E}\kern-.125emX}}
 
\newcommand{\bsf}[1]{\bm{\mathsf{#1}}}
 
 \newcommand{\RGL}{R_{\textnormal{G,L}}}
  \newcommand{\RGU}{R_{\textnormal{G,U}}}
\definecolor{softorange}{RGB}{255,232,204}
\definecolor{softgreen}{RGB}{213,255,218}
\definecolor{softblue}{RGB}{220,232,255}
\definecolor{softgray}{RGB}{238,238,238}
\definecolor{darkorange}{RGB}{190,105,35}
\definecolor{darkgreen}{RGB}{55,130,70}
\definecolor{darkblue}{RGB}{65,95,155}

\begin{document}

\title{Fundamental Limits of Quantized MIMO ISAC under Gaussian Signaling}

\author{
\IEEEauthorblockN{Hossein Atrsaei\textsuperscript{*}, Mireille Sarkiss\textsuperscript{\textdagger}, Mich\`ele Wigger\textsuperscript{\S}}

\IEEEauthorblockA{\textsuperscript{*}LTCI, T\'el\'ecom Paris, Institut Polytechnique de Paris, Palaiseau, France\\
Email: hossein.atrsaei@ip-paris.fr}

\IEEEauthorblockA{\textsuperscript{\textdagger}SAMOVAR, T\'el\'ecom SudParis, Institut Polytechnique de Paris, Palaiseau, France\\
Email: mireille.sarkiss@telecom-sudparis.eu}

\IEEEauthorblockA{\textsuperscript{\S}Universit\'e Paris-Saclay, CNRS, CentraleSup\'elec, Laboratoire des Signaux et Syst\`emes, 91190 Gif-sur-Yvette, France\\
Email: michele.wigger@centralesupelec.fr}
}
\sloppy
\maketitle

\begin{abstract}

We study a quantized multiple-input multiple-output (MIMO) integrated sensing and communication (ISAC) system in which the communication and sensing receivers each apply analog spatial combining followed by scalar subtractive dithered quantization. This quantization model leads to an additive effective-noise representation with non-Gaussian noise. We derive upper and lower bounds on the capacity of this channel. Numerical results show that these bounds are tight at low signal-to-noise ratios (SNR) and saturate at high SNR due to finite-resolution quantization. They also show that, despite the effective noise being  non-Gaussian, independent and identically distributed (i.i.d.) isotropic Gaussian signaling achieves rates close to capacity. Focusing on i.i.d.\ Gaussian signaling, this paper also presents a closed-form expression for the linear minimum mean-squared error (LMMSE) achieved under a Kronecker sensing-channel model. Numerical results show that the LMMSE also saturates at high SNR, where the saturation level increases as the spatial combining ratio decreases, and for combining ratios below one, saturation  occurs even without quantization.

\end{abstract}


\section{Introduction}
\label{sec:intro}

This work studies information-theoretic limits of integrated sensing and communication (ISAC) systems and the tradeoffs between communication and sensing performance in these systems. Previous works have characterized fundamental performance limits of ISAC systems under different channel models and sensing metrics~\cite{zeng2023tradeoff, deng2022limits, Ahmadipour,Bloch,Joudeh}. 

A key practical limitation in large-scale multiple-input multiple-output (MIMO) systems is the power consumption and hardware cost of analog-to-digital converters (ADCs). Low-resolution ADCs, often combined with analog spatial processing prior to quantization, offer an effective way to reduce receiver complexity. The impact of coarse quantization on MIMO communication has been widely studied, including one-bit and low-resolution receiver architectures~\cite{khalili2018mimo,liang2016onebit,choi2016general, singh2009limits,mo2015capacity}. However, the capacity of quantized MIMO channels generally does not admit a closed-form expression, and existing works typically provide bounds or approximations for specific settings, such as one-bit ADCs or asymptotic regimes.

In this paper, we  derive upper and lower bounds on the capacity of a MIMO ISAC system with spatial combiners and quantizers at the receiver side. The focus is on dithered quantization, which admits an additive-noise representation~\cite{atrsaei2026capacity}, but with a non-Gaussian noise. We assume channel state information at the communication receiver (CSIR) and no channel state information at the transmitter (no-CSIT). Our capacity lower bound follows from the worst-case noise property of Gaussian noise, while the upper bound follows from the maximum-entropy property of Gaussian random vectors. Numerical results show that our bounds are tight at low signal-to-noise ratios (SNR) and saturate at high SNR because of the limitations of the quantizers. 
Moreover, for all SNRs independent and identically distributed (i.i.d.) isotropic Gaussian signaling achieves rates close to capacity, despite the fact that  quantization induces non-Gaussian noise. 

We indeed consider a bistatic ISAC system in which the sensing receiver also processes its received signal through linear combiners and quantizers, before estimating the target response. 
For simplicity, a Kronecker model is used for the sensing channel~\cite{weichselberger2006stochastic}.
In the spirit of \emph{communication-centric} ISAC systems, we present a closed-form expression for the LMMSE at the sensing receiver under i.i.d. Gaussian signaling. 
Numerical results for Jakes channel model illustrate that the sensing LMMSE saturates at high SNR due to the quantization noise. They also show that, when spatial dimension reduction is applied (i.e., when the linear combiners map to smaller dimensions), saturation can occur even without quantization.


\textit{Notations:}
Random scalars, vectors, matrices, and sets are denoted by
$\mathsf{a}$, $\bsf{a}$, $\bsf{A}$, and $\mathsf{A}$, respectively, while deterministic quantities are denoted by $a$, $\mathbf{a}$, $\mathbf{A}$, and $\mathcal{A}$. We denote the covariance matrix of a random vector $\bsf{a}$ by $\mathbf{R}_{\bsf{a}}$,
and the Kronecker product and column-wise vectorization operators by $\otimes$ and $\mathrm{vec}(\cdot)$, respectively. 
 

\section{Channel Model and Performance Metrics}
\label{sec:system_model}

\subsection{Channel Model}
\label{subsec:channel_model}

We consider the quantized MIMO ISAC architecture shown in Fig.~\ref{fig:system_model}. A transmitter equipped with $N_t$ antennas sends a dual-functional random waveform
$\bsf{X} = [\bsf{x}_1,\ldots,\bsf{x}_T] \in \mathbb{C}^{N_t \times T}$
over a coherence block of $T \geq N_t$ channel uses, serving simultaneously for communication and sensing. The transmit covariance matrix is $\mathbf{R}_{\bsf{X}} = \mathbb{E}\!\left[\frac{1}{T}\bsf{X}\bsf{X}^H\right]$ with the average power constraint $\mathrm{Tr}(\mathbf{R}_{\bsf{X}}) = N_t P_0$.

The unquantized received signals at the communication and sensing receivers are
\begin{subequations}
\label{eq:unquantized_model}
\begin{IEEEeqnarray}{rCl}
    \bsf{Y}_c &=& \bsf{H}\bsf{X} + \bsf{W}_c,
    \label{eq:unquantized_comm}\\
    \bsf{Y}_s &=& \bsf{G}\bsf{X} + \bsf{W}_s.
    \label{eq:unquantized_sensing}
\end{IEEEeqnarray}
\end{subequations}
Here, $\bsf{H} \in \mathbb{C}^{N_c \times N_t}$ denotes the communication channel, and $\bsf{G} \in \mathbb{C}^{N_s \times N_t}$ denotes the sensing (target-response) channel. The noise matrices have independent circularly symmetric complex Gaussian (CSCG) entries across antennas and channel uses, with $\bsf{w}_c \triangleq \mathrm{vec}(\bsf{W}_c) \sim \mathcal{CN}(\mathbf{0}, \sigma_c^2 \mathbf{I}_{N_cT})$ and $\bsf{w}_s \triangleq \mathrm{vec}(\bsf{W}_s) \sim \mathcal{CN}(\mathbf{0}, \sigma_s^2 \mathbf{I}_{N_sT})$.
The communication channel has i.i.d.\ Rayleigh-fading entries, i.e.,
$\mathrm{vec}(\bsf{H}) \sim \mathcal{CN}(\mathbf{0},\mathbf{I}_{N_tN_c})$.

The sensing channel follows the Kronecker correlation model
\cite{weichselberger2006stochastic}
\begin{equation}
    \bsf{G}
    =
    \mathbf{R}_{A}^{1/2}
    \bsf{G}_{0}
    \bigl(\mathbf{R}_{B}^{1/2}\bigr)^T,
    \qquad
    \mathbf{R}_{\bsf{g}}
    =
    \mathbf{R}_{B}\otimes\mathbf{R}_{A},
    \label{eq:kronecker_sensing_model}
\end{equation}
where $\bsf{g} \triangleq \mathrm{vec}(\bsf{G})$, $\bsf{G}_{0}$ has i.i.d.\
$\mathcal{CN}(0,1)$ entries, $\mathbf{R}_{B}\in\mathbb{C}^{N_t\times N_t}$ is the transmit-side correlation matrix, and $\mathbf{R}_{A}\in\mathbb{C}^{N_s\times N_s}$ is the receive-side correlation matrix. The sensing receiver is assumed to know the transmitted waveform $\bsf{X}$, as in the MIMO ISAC model of \cite{liu2022isac}. The transmitter has no channel state information about $\bsf{H}$, so the input distribution $p_{\bsf{X}}$ cannot depend on the realization of $\bsf{H}$.

\begin{figure}[t]
\centering
\begin{tikzpicture}[
    >=Latex,
    font=\footnotesize,
    node distance=0.7cm,
    block/.style={draw, rectangle, minimum height=5mm, minimum width=12mm, align=center},
    sum/.style={draw, circle, inner sep=0pt, minimum size=4.5mm}
]

\node (x) {$x$};
\node[sum, right=of x]      (sumd) {$+$};
\node[block, right=of sumd] (adc)  {$q_{K}(\cdot)$};
\node[sum, right=of adc]    (subd) {$-$};
\node[right=of subd]        (z) {$Q_K(x)$};

\draw[->] (x)    -- (sumd);
\draw[->] (sumd) -- (adc);
\draw[->] (adc)  -- (subd);
\draw[->] (subd) -- (z);

\node[below=0.5cm of sumd] (d1) {$\mathsf{d}$};
\node[below=0.5cm of subd] (d2) {$\mathsf{d}$};
\draw[->] (d1) -- (sumd);
\draw[->] (d2) -- (subd);

\node[draw, dashed, inner sep=5pt,
      fit=(sumd)(adc)(subd)(d1)(d2)] (qbox) {};
\node[below=3pt of qbox, font=\scriptsize]
      {Subtractive Dithered Quantizer};

\end{tikzpicture}
\caption{Subtractive dithered quantizer.}
\label{fig:dithered_quantizer}
\end{figure}
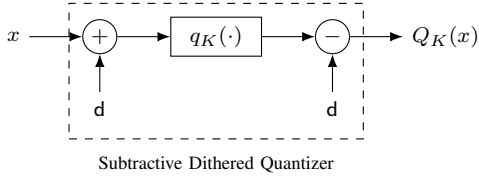
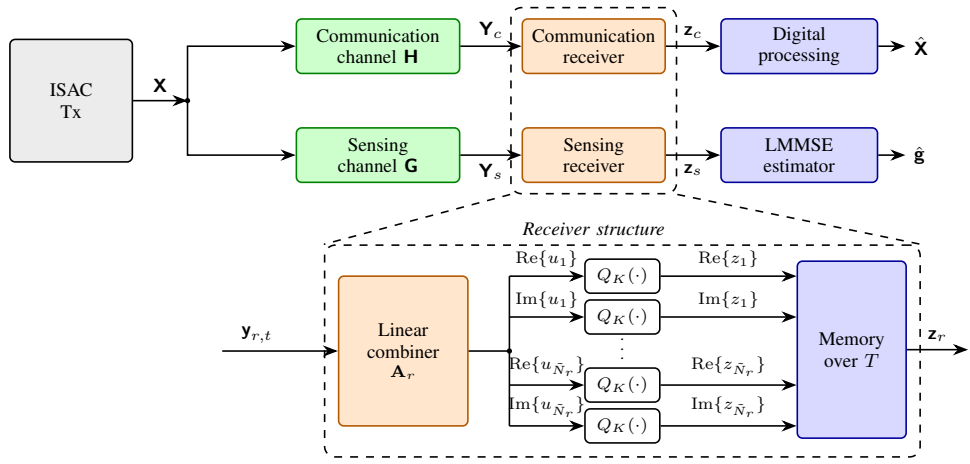
\begin{figure*}[!t]
\centering
\begin{tikzpicture}[
    scale=0.90,
    every node/.style={transform shape},
    >=Stealth,
    semithick,
    font=\footnotesize,
    block/.style={draw,rounded corners=2pt,align=center,minimum height=7mm,minimum width=16mm},
    smallblock/.style={draw,rounded corners=2pt,align=center,minimum height=5mm,minimum width=11mm,font=\scriptsize},
    rxblock/.style={draw,rounded corners=2pt,align=center,minimum height=8mm,minimum width=21mm,fill=orange!20,draw=orange!70!black},
    chan/.style={draw,rounded corners=2pt,align=center,minimum height=8mm,minimum width=24mm,fill=green!20,draw=green!50!black},
    memblock/.style={draw,rounded corners=2pt,align=center,minimum height=26mm,minimum width=16mm,fill=blue!15,draw=blue!60!black},
    dspblock/.style={draw,rounded corners=2pt,align=center,minimum height=8mm,minimum width=23mm,fill=blue!15,draw=blue!60!black},
    dashedbox/.style={draw,dashed,rounded corners=4pt}
]


\node[block, fill=gray!15, minimum width=18mm, minimum height=18mm] (tx)
{ISAC\\Tx};

\coordinate (fork) at ([xshift=8mm]tx.east);

\node[chan] (chc) at ([xshift=28mm,yshift=8mm]fork)
{Communication\\channel $\bsf{H}$};

\node[rxblock, right=9mm of chc] (rxc)
{Communication\\receiver};

\node[chan] (chs) at ([xshift=28mm,yshift=-8mm]fork)
{Sensing\\channel $\bsf{G}$};

\node[rxblock, right=9mm of chs] (rxs)
{Sensing\\receiver};

\node[dspblock, right=8mm of rxc] (dspc)
{Digital\\processing};

\node[dspblock, right=8mm of rxs] (dsps)
{LMMSE\\estimator};

\node[right=4mm of dspc] (xhat)
{$\hat{\bsf{X}}$};

\node[right=4mm of dsps] (ghat)
{$\hat{\bsf{g}}$};

\draw[->] (tx.east)
--
node[above,pos=0.5] {$\bsf{X}$}
(fork);

\fill (fork) circle (1.1pt);

\draw[->] (fork) |- (chc.west);
\draw[->] (fork) |- (chs.west);

\draw[->] (chc)
--
node[above] {$\bsf{Y}_c$}
(rxc);

\draw[->] (chs)
--
node[below] {$\bsf{Y}_s$}
(rxs);

\draw[->] (rxc)
--
node[above] {$\bsf{z}_c$}
(dspc);

\draw[->] (rxs)
--
node[below] {$\bsf{z}_s$}
(dsps);

\draw[->] (dspc) -- (xhat);
\draw[->] (dsps) -- (ghat);

\node[dashedbox, fit=(rxc)(rxs), inner sep=4pt] (zoom) {};


\node[font=\footnotesize\itshape] at ([yshift=-5mm]zoom.south)
{Receiver structure};

\node[
    block,
    fill=orange!20,
    draw=orange!70!black,
    minimum width=19mm,
    minimum height=22mm
] (A) at ([xshift=-28mm,yshift=-23mm]zoom.south)
{Linear\\combiner\\[-1pt]{\scriptsize $\mathbf A_r$}};

\node[smallblock, right=17mm of A, yshift=11mm] (q1R)
{$Q_K(\cdot)$};

\node[smallblock, right=17mm of A, yshift=5mm] (q1I)
{$Q_K(\cdot)$};

\node[smallblock, right=17mm of A, yshift=-5mm] (qNR)
{$Q_K(\cdot)$};

\node[smallblock, right=17mm of A, yshift=-11mm] (qNI)
{$Q_K(\cdot)$};

\node[memblock, right=20mm of q1I, yshift=-4.8mm] (mem)
{Memory\\over $T$};

\coordinate (bus) at ([xshift=6mm]A.east);

\draw[->] ([xshift=-17mm]A.west)
--
node[above,pos=0.30] {$\bsf{y}_{r,t}$}
(A.west);

\draw (A.east) -- (bus);
\fill (bus) circle (1.1pt);

\draw (bus |- q1R.west) -- (bus |- qNI.west);

\draw[->] (bus |- q1R.west)
--
node[above,font=\scriptsize] {$\Re\{u_1\}$}
(q1R.west);

\draw[->] (bus |- q1I.west)
--
node[above,font=\scriptsize] {$\Im\{u_1\}$}
(q1I.west);

\draw[->] (q1R.east)
--
node[above,font=\scriptsize] {$\Re\{z_1\}$}
(q1R.east -| mem.west);

\draw[->] (q1I.east)
--
node[above,font=\scriptsize] {$\Im\{z_1\}$}
(q1I.east -| mem.west);

\node[font=\scriptsize] at ($(q1I)!0.365!(qNR)$) {$\vdots$};

\draw[->] (bus |- qNR.west)
--
node[above,font=\scriptsize] {$\Re\{u_{\tilde N_r}\}$}
(qNR.west);

\draw[->] (bus |- qNI.west)
--
node[above,font=\scriptsize] {$\Im\{u_{\tilde N_r}\}$}
(qNI.west);

\draw[->] (qNR.east)
--
node[above,font=\scriptsize] {$\Re\{z_{\tilde N_r}\}$}
(qNR.east -| mem.west);

\draw[->] (qNI.east)
--
node[above,font=\scriptsize] {$\Im\{z_{\tilde N_r}\}$}
(qNI.east -| mem.west);

\draw[->] (mem.east)
--
node[above] {$\bsf{z}_r$}
++(9mm,0);

\node[dashedbox, fit=(A)(q1R)(q1I)(qNR)(qNI)(mem), inner sep=5pt] (detail) {};

\draw[dashed] (zoom.south west) -- (detail.north west);
\draw[dashed] (zoom.south east) -- (detail.north east);

\end{tikzpicture}

\caption{Quantized MIMO-ISAC system model. The transmitted random waveform
$\bsf{X}\in\mathbb{C}^{N_t\times T}$ is observed as
$\bsf{Y}_c=\bsf{H}\bsf{X}+\bsf{W}_c$ at the communication receiver and as
$\bsf{Y}_s=\bsf{G}\bsf{X}+\bsf{W}_s$ at the sensing receiver. Both receivers
employ the same quantization architecture
of~\cite{ruan2024taskbased,shlezinger2019asymptotic,shlezinger2019hardware}:
for each snapshot, the received vector $\bsf{y}_{r,t}$ is linearly combined
by $\mathbf A_r$, each real and imaginary component of the resulting complex
entries is quantized by a scalar subtractive dithered quantizer $Q_K(\cdot)$,
and the quantized outputs are accumulated in memory over the coherence block.}
\label{fig:system_model}
\end{figure*}

Each receiver adopts the hardware-limited quantization architecture of
\cite{shlezinger2019asymptotic,shlezinger2019hardware,ruan2024taskbased}, consisting of an analog spatial combiner followed by parallel scalar ADCs. For each snapshot $t\in\{1,\ldots,T\}$ and each receiver $r\in\{c,s\}$, the received vector $\bsf{y}_{r,t}\in\mathbb{C}^{N_r}$ is processed by a deterministic analog combiner
$\mathbf{A}_r\in\mathbb{C}^{\tilde N_r\times N_r}$ as
\begin{equation}
    \bsf{u}_{r,t}
    =
    \mathbf{A}_r\bsf{y}_{r,t},
    \qquad
    r\in\{c,s\},\quad t=1,\ldots,T.
    \label{eq:snapshot_combining}
\end{equation}
The combiner is constrained to be semi-unitary,
\begin{equation}
    \mathbf{A}_r\mathbf{A}_r^H
    =
    \mathbf{I}_{\tilde N_r},
    \qquad
    r\in\{c,s\},
    \label{eq:semi_unitary_combiner}
\end{equation}
which prevents noise amplification and preserves the spatial whiteness of the receiver noise after combining.

The real and imaginary parts of each entry of $\bsf{u}_{r,t}$ are quantized independently using a $K$-level subtractive dithered uniform quantizer, shown in Fig.~\ref{fig:dithered_quantizer}. For a real input $x$, the quantizer output is
\begin{equation}
    Q_K(x)
    =
    q_K(x+\mathsf{d})-\mathsf{d},
    \label{eq:subtractive_quantizer}
\end{equation}
where $\mathsf{d}$ is an independent dither uniformly distributed over
$[-\Delta_r/2,\Delta_r/2]$. The map $q_K(\cdot)$ is a $K$-level uniform scalar quantizer with dynamic range $[-\gamma_r,\gamma_r]$ and step size $\Delta_r = 2\gamma_r/{K}$.
The dynamic range is partitioned into $K$ subintervals of width $\Delta_r$, whose midpoints are the reconstruction levels. Inputs outside $[-\gamma_r,\gamma_r]$ are saturated to the nearest reconstruction level, namely $\pm(\gamma_r-\Delta_r/2)$.

Under Schuchman's conditions \cite{schuchman1964dither,gray1993dithered}, and assuming negligible overload, the subtractive dithered quantizer yields an additive-noise representation in which the quantization error is statistically independent of the input
\cite{atrsaei2026capacity}. Thus, for each receiver $r\in\{c,s\}$ and snapshot $t$,
\begin{equation}
    \bsf{z}_{r,t}
    =
    \bsf{u}_{r,t}
    +
    \bsf{w}_{q,r,t},
    \label{eq:additive_quantization_model}
\end{equation}
where the entries of $\bsf{w}_{q,r,t}$ are independent, with independent real and imaginary parts uniformly distributed over $[-\Delta_r/2,\Delta_r/2]$. Consequently, 
\begin{equation}
    \mathbb{E}\!\left[|(\bsf{w}_{q,r,t})_i|^2\right]
    =
    \frac{\Delta_r^2}{6}.
\end{equation}

The dynamic range $\gamma_r$ must be chosen large enough so that the dithered input remains within $[-\gamma_r,\gamma_r]$ with high probability. For notational compactness, we stack the samples over one coherence block as $\bsf{u}_r = \mathrm{vec}\bigl([\bsf{u}_{r,1},\ldots,\bsf{u}_{r,T}]\bigr)$, and define $\bsf{d}_r$, $\bsf{w}_{q,r}$, $\bsf{z}_r$, and the other block-level vectors analogously. Following
\cite{shlezinger2019asymptotic,shlezinger2019hardware,ruan2024taskbased}, we set
\begin{equation}
    \gamma_r
    =
    \eta
    \sqrt{
    \max_{i=1,\ldots,\tilde N_rT}
    \mathbb{E}\!\left[
        \left|(\bsf{u}_r+\bsf{d}_r)_i\right|^2
    \right]},
    \qquad
    r\in\{c,s\},
    \label{eq:gamma_definition}
\end{equation}
where $\eta>0$ controls the overload probability. Since
\begin{align}
    \mathbb{E}\!\left[
        \left|(\bsf{u}_r+\bsf{d}_r)_i\right|^2
    \right]
    &=
    \left(\mathbf{R}_{\bsf{u}_r}\right)_{i,i}
    +
    \frac{\Delta_r^2}{6},
\end{align}
and since $\Delta_r=2\gamma_r/K$, we obtain
\begin{equation}
    \gamma_r^2
    =
    \kappa
    \max_{i=1,\ldots,\tilde N_rT}
    \left(\mathbf{R}_{\bsf{u}_r}\right)_{i,i},
    \label{eq:dynamic_range}
\end{equation}
where
\begin{equation}
    \kappa
    \triangleq
    \eta^2
    \left(
        1-\frac{2\eta^2}{3K^2}
    \right)^{-1}.
    \label{eq:kappa_def}
\end{equation}
This expression is well defined when $\eta < \sqrt{3/2}\,K$.

The stacked pre-quantization and quantized signals are
\begin{equation}
    \bsf{u}_r
    =
    (\mathbf{I}_T\otimes\mathbf{A}_r)\bsf{y}_r,
    \qquad
    \bsf{z}_r
    =
    \bsf{u}_r+\bsf{w}_{q,r},
    \label{eq:stacked_receiver_model}
\end{equation}
where $\bsf{y}_r\triangleq\mathrm{vec}(\bsf{Y}_r)$. Substituting
\eqref{eq:unquantized_model} into \eqref{eq:stacked_receiver_model} gives the quantized sensing and communication models
\begin{subequations}
\label{eq:quantized_models}
\begin{align}
    \bsf{z}_s
    &=
    (\bsf{X}^T\otimes\mathbf{A}_s)\bsf{g}
    +
    \bar{\bsf{w}}_s,
    \label{eq:quantized_sensing_model}\\
    \bsf{z}_c
    &=
    (\mathbf{I}_T\otimes\mathbf{A}_c\bsf{H})\bsf{x}
    +
    \bar{\bsf{w}}_c,
    \label{eq:quantized_comm_model}
\end{align}
\end{subequations}
where $\bsf{x}\triangleq\mathrm{vec}(\bsf{X})$ and
\begin{equation}
    \bar{\bsf{w}}_r
    \triangleq
    (\mathbf{I}_T\otimes\mathbf{A}_r)\bsf{w}_r
    +
    \bsf{w}_{q,r},
    \qquad
    r\in\{c,s\},
    \label{eq:effective_noise_vec}
\end{equation}
is the effective noise vector. The noise components $\bsf{w}_{r}$ and $\bsf{w}_{q,r}$ are mutually independent and have i.i.d. entries with variances $\sigma_{r}^2$ and $\Delta_r^2/6=2 \gamma_r^2/(3K^2)$, respectively.
Combining these observations with \eqref{eq:dynamic_range} and the semi-unitary constraint \eqref{eq:semi_unitary_combiner}, the effective noise covariance becomes
\begin{equation}
    \mathbf{R}_{\bar{\bsf{w}}_r}
    =
    \sigma_{0,r}^{2}\mathbf{I}_{\tilde N_rT},
    \qquad
    r\in\{c,s\},
    \label{eq:effective_noise_covariance}
\end{equation}
where, for each $r\in\{c,s\}$,
\begin{equation}
    \sigma_{0,r}^{2}
    \triangleq
    \sigma_r^2
    +
    \beta
    \max_{i=1,\ldots,\tilde N_rT}
    \left(\mathbf{R}_{\bsf{u}_r}\right)_{i,i},
    \qquad
    \beta
    \triangleq
    \frac{2\kappa}{3K^2}.
    \label{eq:effective_noise_variance}
\end{equation}

\subsection{Performance Metrics}
\label{subsec:performance_metrics}

\subsubsection{Sensing distortion}

The sensing performance is measured by the normalized LMMSE in estimating the target-response vector $\bsf{g}$. Since the sensing receiver knows the realization of the transmitted waveform $\bsf{X}=\mathbf{X}$, the conditional sensing model is linear in $\bsf{g}$. The LMMSE estimator is therefore
\cite{kay1993estimation}
\begin{equation}
    \hat{\bsf{g}}
    =
    \mathbf{R}_{\bsf{g}}
    \tilde{\mathbf{X}}^H
    \bigl(
        \tilde{\mathbf{X}}
        \mathbf{R}_{\bsf{g}}
        \tilde{\mathbf{X}}^H
        +
        \mathbf{R}_{\bar{\bsf{w}}_s}
    \bigr)^{-1}
    \bsf{z}_s,
    \label{eq:lmmse_estimator}
\end{equation}
where
\begin{equation}
    \tilde{\mathbf{X}}
    \triangleq
    \mathbf{X}^T\otimes\mathbf{A}_s.
\end{equation}
The corresponding LMMSE error covariance matrix is
\begin{equation}
    \mathbf{R}_{\bsf{g}\mid\bsf{z}_s,\bsf{X}=\mathbf{X}}
    =
    \mathbf{R}_{\bsf{g}}
    -
    \mathbf{R}_{\bsf{g}}
    \tilde{\mathbf{X}}^H
    \bigl(
        \tilde{\mathbf{X}}
        \mathbf{R}_{\bsf{g}}
        \tilde{\mathbf{X}}^H
        +
        \mathbf{R}_{\bar{\bsf{w}}_s}
    \bigr)^{-1}
    \tilde{\mathbf{X}}
    \mathbf{R}_{\bsf{g}}.
    \label{eq:error_cov}
\end{equation}
The conditional LMMSE is given by the trace of \eqref{eq:error_cov}, which, by the Woodbury matrix identity, simplifies to:
\begin{equation}
\begin{aligned}
    &\sigma_{\bsf{g}\mid\bsf{X}=\mathbf{X}}^2(\mathbf{A}_s)\\
    &=
    \frac{1}{N_tN_s}
    \operatorname{Tr}\!\Biggl(
    \biggl[
        \mathbf{R}_{\bsf{g}}^{-1}
        +
        \frac{1}{\sigma_{0,s}^{2}}
        \bigl(
            \mathbf{X}^*\mathbf{X}^T
            \otimes
            \mathbf{A}_s^H\mathbf{A}_s
        \bigr)
    \biggr]^{-1}
    \Biggr).
\end{aligned}
\label{eq:sigma_g_inv}
\end{equation}
For a given input distribution $p_{\bsf{X}}$, the sensing distortion is obtained by minimizing the expected conditional LMMSE distortion over all admissible sensing combiners:
\begin{equation}
    \epsilon(p_{\bsf{X}})
    \triangleq
    \min_{\mathbf{A}_s:\, \mathbf{A}_s\mathbf{A}_s^H=\mathbf{I}_{\tilde N_s}}
    \mathbb{E}_{\bsf{X}}\!\left[
        \sigma_{\bsf{g}\mid\bsf{X}}^2(\mathbf{A}_s)
    \right].
    \label{eq:sensing_distortion}
\end{equation}

\subsubsection{Communication rate}

Since the communication receiver has perfect channel state information, for a given input distribution $p_{\bsf{X}}$, the following rate is achievable: 
\begin{equation}
    R(p_{\bsf{X}})
    \triangleq
    \frac{1}{T}
    \max_{\mathbf{A}_c:\, \mathbf{A}_c\mathbf{A}_c^H=\mathbf{I}_{\tilde N_c}}
    I(\bsf{X};\bsf{Z}_c\mid\bsf{H}),
    \label{eq:rate_px}
\end{equation}
where $\bsf{Z}_c
    \triangleq
    [\bsf{z}_{c,1},\ldots,\bsf{z}_{c,T}]$.

\section{Performance under Gaussian Signaling}
\label{sec:gaussian_performance}

We consider i.i.d.\ isotropic Gaussian signaling, i.e., the transmitted vectors are independent across time, with
\begin{equation}
    \bsf{x}_t
    \sim
    \mathcal{CN}(\mathbf{0},P_0\mathbf{I}_{N_t}),
    \qquad
    t=1,\ldots,T.
    \label{eq:gaussian_input}
\end{equation}
Thus, $\mathbf{R}_{\bsf{X}}=P_0\mathbf{I}_{N_t}$. For this input distribution, we denote the sensing distortion $\epsilon(p_{\bsf{X}})$ and the communication rate $R(p_{\bsf{X}})$ by $\epsilon_G$ and $R_G$, respectively.
In the following, we characterize $\epsilon_G$ and provide bounds on $R_G$, which is difficult to characterize exactly due to the non-Gaussian effective noise. 

\subsection{Bounds on $R_G$}
\label{subsec:RG_bounds}

As shown in Appendix~\ref{App:sigma2},  under isotropic Gaussian signaling, irrespective of the choice of the linear combiner
and the channel $\mathbf{H}$,
\begin{equation}
    \sigma_{0,c}^{2}
    =
    \sigma_c^2+\beta(\sigma_c^2+N_tP_0).
    \label{eq:sigma0_nocsit}
\end{equation}
Let
$\lambda_1(\bsf{H}\bsf{H}^H)\geq\cdots\geq\lambda_{N_c}(\bsf{H}\bsf{H}^H)$
denote the ordered eigenvalues of $\bsf{H}\bsf{H}^H$, and define
\begin{equation}
    \RGL
    \triangleq
    \mathbb{E}_{\bsf{H}}\!\left[
    \sum_{i=1}^{\tilde{N}_c}
    \log\!\left(
        1+\frac{P_0\lambda_i(\bsf{H}\bsf{H}^H)}{\sigma_{0,c}^{2}}
    \right)
    \right],
    \label{eq:cg_nocsit}
\end{equation}
and
\begin{equation}
    \RGU
    \triangleq
    \RGL
    +
    \tilde{N}_c\log(\pi e\sigma_{0,c}^{2})
    -
    2\tilde{N}_c\,h\!\left(\Re(\mathsf{\bar{W}}_{c,i})\right),
    \label{eq:RGU}
\end{equation}
where $\Re(\mathsf{\bar{W}}_{c})$ denotes the real part of one scalar component of the effective communication noise. Its density is
\begin{equation}
    f_{\Re(\mathsf{\bar{W}}_c)}(w)
    =
    \frac{1}{\Delta_c}
    \left[
    Q\!\left(
        \frac{w-\Delta_c/2}{\sigma_c/\sqrt{2}}
    \right)
    -
    Q\!\left(
        \frac{w+\Delta_c/2}{\sigma_c/\sqrt{2}}
    \right)
    \right],
    \label{eq:real_eff_noise_pdf}
\end{equation}
with $Q(\cdot)$ denoting the standard Gaussian $Q$-function.

Notice that the lower bound  $R_{\textnormal{G,L}}$ only captures the dominant $\tilde{N}_c$ eigenmodes of $(\bsf{H}\bsf{H}^H)$, but not the smaller eigenmodes. 

\begin{theorem}[Bounds on $R_G$]
\label{thm:rate_bounds}
Under Gaussian signaling, the achievable rate satisfies 
\begin{equation}
    \RGL
    \leq
    R_{\textnormal{G}}
    \leq
    \RGU.
    \label{eq:rate_sandwich}
\end{equation}
\end{theorem}

\begin{IEEEproof}
See Appendix~\ref{app:proof_rate_bounds}. 
\end{IEEEproof}

\begin{remark}
The upper bound in Theorem~\ref{thm:rate_bounds} is also an upper bound on  capacity. As a consequence, when the two bounds coincide, isotropic Gaussian signaling is capacity-achieving. Our numerical results in Section~\ref{sec:numerical_results} show that this is indeed the case  in the low and moderate SNR regimes.
\end{remark}

The bounds in \eqref{eq:rate_sandwich} admit simple interpretations in several asymptotic regimes. We express these regimes in terms of 
\begin{equation}
    \mathrm{SNR}
    \triangleq
    \frac{N_tP_0}{\sigma_c^2}.
    \label{eq:SNR_def}
\end{equation}

\textit{Low SNR.}
For $P_0/\sigma_c^2\to 0$, we have
$\sigma_{0,c}^{2}=\sigma_c^2(1+\beta)+o(1)$. Using
$\log(1+x)=x+o(x)$ gives
\begin{equation}
    \RGL
    =
    \frac{P_0}{\sigma_c^2(1+\beta)}
    \mathbb{E}_{\bsf{H}}\!\left[
        \sum_{i=1}^{\tilde{N}_c}
        \lambda_i(\bsf{H}\bsf{H}^H)
    \right]
    +
    o(P_0).
    \label{eq:low_snr_asymptotic}
\end{equation}
Thus, quantization induces a multiplicative penalty $1+\beta$ on the low-SNR slope.

\textit{High SNR with fixed resolution.}
For $P_0/\sigma_c^2\to\infty$ with fixed $K$, the effective-noise variance grows linearly with $P_0$:
\begin{equation}
    \sigma_{0,c}^{2}
    =
    \beta N_tP_0+o(P_0).
\end{equation}
Consequently, the lower bound saturates as
\begin{equation}
    \lim_{P_0\to\infty}
    \RGL
    =
    \mathbb{E}_{\bsf{H}}\!\left[
    \sum_{i=1}^{\tilde{N}_c}
    \log\!\left(
        1+\frac{\lambda_i(\bsf{H}\bsf{H}^H)}{\beta N_t}
    \right)
    \right].
    \label{eq:high_snr_asymptotic}
\end{equation}

\textit{High resolution.}
As $\Delta_c\to 0$, the quantization noise vanishes and
\begin{equation}
    h\!\left(\Re(\mathsf{\bar{W}}_c)\right)
    \to
    \frac{1}{2}\log(\pi e\sigma_c^2).
\end{equation}
Therefore, the gap $\RGU-\RGL$ tends to zero, and the lower and upper bounds become asymptotically tight. In this regime, both bounds converge to the Gaussian MIMO capacity without quantization.

\textit{Massive MIMO.}
By the favorable-propagation and channel-hardening properties of i.i.d.\ Rayleigh massive MIMO channels \cite{bjornson2017massive}, as $N_t/\tilde{N}_c\to\infty$,
\begin{equation}
    \RGL
    \longrightarrow
    \tilde{N}_c
    \log\!\left(
        1+\frac{N_tP_0}{\sigma_{0,c}^{2}}
    \right).
    \label{eq:massive_mimo_asymptotic_capacity}
\end{equation}

\subsection{Closed-Form Expression for $\epsilon_G$}
\label{subsec:eps_G_closed_form}

Adapting \cite[Appendix~A]{ruan2024taskbased} to the present model, we obtain:

\begin{theorem}[Closed-form expression for $\epsilon_G$]
\label{pro:pcs_sensing_bound}
Under i.i.d.\ isotropic Gaussian signaling,  i.e., $\mathbf{R}_{\bsf{X}}^*= P_0 \mathbf I_{N_t}$,\footnote{The expression in
\eqref{eq:epsilon_G_general} holds
for any feasible input covariance $\mathbf{R}_{\bsf{X}}\succeq\mathbf{0}$
satisfying $\operatorname{Tr}(\mathbf{R}_{\bsf{X}})=N_tP_0$.}
\begin{IEEEeqnarray}{rCl}
    \epsilon_G
    &= &\frac{1}{N_tN_s}\operatorname{Tr}(\mathbf{R}_{B}\otimes\mathbf{R}_{A})\nonumber\\
  &  &- \frac{1}{N_tN_s}\max_{\{\sigma_i^2\geq0\}}
    \sum_{n=1}^{N_t}\sum_{i=1}^{\tilde{N}_s}
    \frac{\operatorname{Tr}(\mathbf{R}_{B})}{N_t} \nonumber \\
    &&\; \cdot \mathbb{E}_{\bsf{X}}\!\Bigg[
    \frac{\lambda_{A,i}\,\lambda'_{n}\,\sigma_i^2}
    {\lambda'_{n}\sigma_i^2 + \frac{\beta}{\tilde{N}_s}
        \operatorname{Tr}(\mathbf{R}_{\bsf{X}}^{*}\mathbf{R}_{B})
        \sum_{j=1}^{\tilde{N}_s}\sigma_j^2 + \sigma_s^2(1+\beta)}\Bigg],\nonumber\\
    \label{eq:epsilon_G_general}
\end{IEEEeqnarray}
where $\{\lambda_{A,i}\}$ and $\{\lambda'_n\}$ are the eigenvalues of 
$\mathbf{R}_{A}$ and $\mathbf{R}_{B}^{1/2}\bsf{X}^{*}\bsf{X}^{T}(\mathbf{R}_{B}^{1/2})^H$
 in descending order, and maximization is subject to
\begin{equation}
    \sum_{i=N_s-\tilde{N}_s+1}^{N_s}\lambda_{A,i}
    \;\leq\;
    \sum_{i=1}^{\tilde{N}_s}\sigma_i^2
    \;\leq\;
    \sum_{i=1}^{\tilde{N}_s}\lambda_{A,i}.\label{eq:pcs_sensing_sigma_constraint}
\end{equation}

\end{theorem}

\begin{IEEEproof}
See Appendix~\ref{app:proof_pcs_sensing_bound}.
\end{IEEEproof}

\section{Numerical Results}
\label{sec:numerical_results}

Unless stated otherwise, the transmitter has $N_t=16$ antennas and each receiver has $16$ antennas, i.e., $N_c=N_s=16$.  The coherence  length is $T=40$ channel uses, and the noise variances are $\sigma_c^2=\sigma_s^2=10^{-3}$. The dither overload factor is set to $\eta=1$.

We define the combining ratios as 
\begin{equation}
   r_{r}\triangleq \tilde N_r/N_r, \quad r\in\{c,s\},
\end{equation} 
and use the SNR defined in \eqref{eq:SNR_def}; since $\sigma_c^2=\sigma_s^2$ here, it coincides with $N_tP_0/\sigma_s^2$ for the sensing plots.

The sensing channel is modeled with Jakes correlations~\cite{jakes1994microwave}
\[
\begin{aligned}
    (\mathbf{R}_A)_{n_1,n_2}
    &= J_0\!\left(\pi |n_1-n_2|\right),\\
    (\mathbf{R}_B)_{n_1,n_2}
    &= J_0\!\left(0.8\pi |n_1-n_2|\right),
\end{aligned}
\]
where $J_0(\cdot)$ is the zero-order Bessel function of the first kind.

Expectation over  $\bsf{H}$ in \eqref{eq:cg_nocsit}  is calculated using a Monte Carlo method, and the entropy term $h(\textnormal{Re}(\mathsf{\bar{W}}_c))$ in  \eqref{eq:RGU} is computed by numerical integration of the density in \eqref{eq:real_eff_noise_pdf}.

Fig.~\ref{fig:R1} illustrates the bounds of
Theorem~\ref{thm:rate_bounds} on the Gaussian-signaling rate $R_G$ versus SNR, as well as the unquantized no-combiner reference $\mathbb{E}_{\bsf H}\!\big[\sum_i\log \big(1+P_0\lambda_i(\bsf H\bsf H^H)/\sigma_c^2\big)\big]$, for combining ratio $r_c=1$. For this value, the optimal combiner is the identity and the Gaussian rate $R_{\textnormal{G}}$ can be computed numerically. We observe that at low and moderate SNRs, the lower bound $\RGL$, the exact rate $R_G$, and the upper bound $\RGU$ are indistinguishably close. This holds  because the Gaussian noise dominates the effective noise. Since $\RGU$ is also an upper bound on capacity, we can also deduce that i.i.d. isotropic Gaussian signaling operates close to capacity in these regimes. At high SNR, a  moderate gap appears between both bounds and the exact $R_{\textnormal{G}}$, which is due to the non-Gaussianity of the effective noise induced by quantization. Both the Gaussian rate $R_{\textnormal{G}}$ and capacity saturate in the high SNR regime, because the quantization-noise variance increases with the transmit power. Finally, compared with the unquantized no-combiner benchmark, we deduce that at very low SNR, quantization induces only a negligible rate penalty, in particular if the quantizer resolution $K$ is sufficiently large ($K=4$).

\begin{figure}[t]
  \centering
  \includegraphics[width=.95\columnwidth]{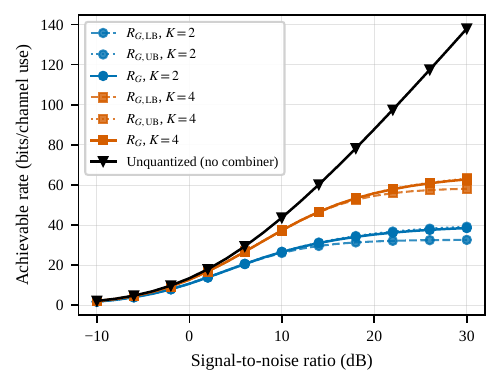}
  \caption{Lower and upper bounds, exact rate, and unquantized no-combiner reference plotted versus SNR for ADC resolutions $K\in\{2,4\}$ ($N_t=16$, $r_c=1$, and $T=40$).}
  \label{fig:R1}
\end{figure}

Fig.~\ref{fig:R2} shows the effect of the communication combining ratio $r_c$ on the rate, plotting the bounds of Theorem~\ref{thm:rate_bounds} versus SNR for $r_c\in\{0.5,0.75,1\}$ at fixed resolution $K=2$. For $r_c=0.5$ and $r_c=0.75$, only a fraction of the $16$ spatial dimensions are retained before quantization, whereas for $r_c=1$ all dimensions are kept and the exact rate $R_G$ can again be computed. At low SNR all combining ratios coincide, since in this regime the rate is noise-limited rather than dimension-limited. As the SNR increases, the impact of the combining ratio becomes more pronounced. The curves for $r_c=0.75$ and $r_c=1$ remain nearly indistinguishable up to the moderate-SNR regime, whereas the more substantial dimensionality reduction with $r_c=0.5$ starts to show a visible rate loss earlier 

\begin{figure}[t]
  \centering
  \includegraphics[width=.95\columnwidth]{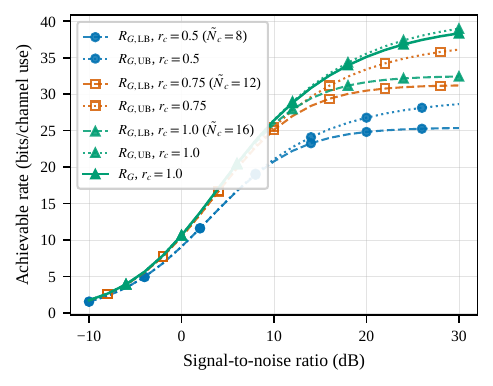}
  \caption{Lower and upper bounds, exact rate (for $r_c=1$), and unquantized no-combiner reference plotted versus SNR for combining ratios $r_c\in\{0.5, 0.75,1\}$ ($N_t=16$, $K=2$, $T=40$).}
  \label{fig:R2}
\end{figure}

{Fig.~\ref{fig:L2} compares the Gaussian-signaling LMMSE $\epsilon_G$ with the unquantized distortion obtained with the optimal linear combiner under the  Gaussian signaling. 
In the expression for $\epsilon_G$ in \eqref{eq:epsilon_G_general}, the
expectation has no closed form for a general $\mathbf{R}_B$. We therefore
approximate it by sample-average approximation (SAA): we draw independent
waveform samples according to \eqref{eq:gaussian_input}, compute the eigenvalues
of
$\mathbf{R}_{B}^{1/2}\bsf{X}^{*}\bsf{X}^{T}(\mathbf{R}_{B}^{1/2})^H$
for each sample, and replace $\mathbb{E}_{\bsf{X}}[\cdot]$ in
\eqref{eq:epsilon_G_general} by the corresponding empirical average. The
resulting finite-sample optimization over
$\{\sigma_i\}_{i=1}^{\tilde N_s}$ is then solved numerically under
\eqref{eq:pcs_sensing_sigma_constraint}.

\begin{figure}[ht!]
  \centering
  \includegraphics[width=0.95\columnwidth]{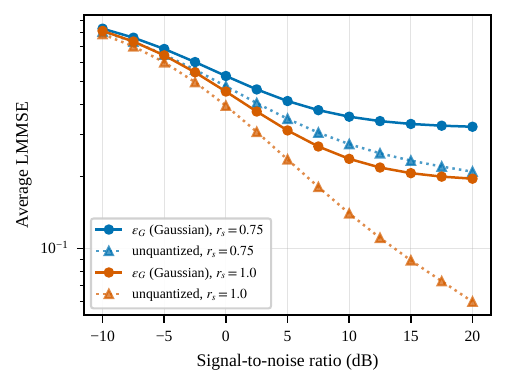}
  \caption{Gaussian-signaling distortion $\epsilon_G$ versus SNR for the Jakes channel ($N_t=16$, $K=2$, $r_s\in\{0.75,1\}$), with the unquantized reference.}
  \label{fig:L2}
\end{figure}

From Fig.~\ref{fig:L2}, we observe that $\epsilon_G$ saturates at high SNR.
Thus, quantization induces saturation not only in the communication rate, but
also in the sensing LMMSE. In contrast, the unquantized benchmark tends to zero
when $r_s=1$, since no spatial dimension is discarded. However, when
$r_s=0.75$, the unquantized benchmark also saturates because the linear combiner
reduces the dimension of the sensing observations before estimation. In fact,
the unquantized curve with $r_s=0.75$ saturates at a higher LMMSE level than the
quantized curve with $r_s=1$. This indicates that, once dimension reduction is present, the dominant sensing loss may come from discarding spatial observations rather than from quantization itself. 

\section{Conclusion and Future Work}
\label{sec:conclusion}

We studied a quantized MIMO ISAC system in which both the communication and sensing receivers employ analog spatial combining followed by scalar subtractive dithered quantization. Focusing on i.i.d.\ isotropic Gaussian signaling under no-CSIT and CSIR, we characterized the LMMSE sensing distortion using optimization and majorization tools, and derived information-theoretic lower and upper bounds on the corresponding achievable communication rate. The proposed upper bound is also an upper bound on capacity.
 
Numerical results suggest that the proposed lower and upper bounds on the Gaussian-input communication rate are tight at low and moderate SNRs, and thus Gaussian signaling is near-optimal. At high SNR, the achievable communication rates saturate at a level determined by the number of quantization levels $K$.  
 
For sensing, we derived a closed-form expression for the LMMSE under i.i.d.\ Gaussian signaling and a Kronecker sensing-channel model. Numerical results show that the LMMSE saturates at high SNR, whereas at low SNR it remains close to the corresponding unquantized LMMSE with the same spatial combiner. The saturation level increases when the sensing combining ratio is reduced, and for combining ratios below one, saturation can occur even without quantization. This shows that both quantization and spatial dimension reduction can fundamentally limit the sensing accuracy, whereas dimension reduction has only a minor effect on the achievable rate.
 
Future work includes tightening the capacity bounds, particularly in the high-SNR regime; see also \cite{atrsaei2026capacity}. When no spatial dimension reduction is applied at the communication receiver, the exact isotropic Gaussian-input rate can be evaluated numerically, revealing a gap to both bounds. A further direction is to characterize the full achievable rate--LMMSE region, beyond the communication-centric point considered here. This includes determining the sensing-optimal performance and its gain over the isotropic Gaussian-signaling LMMSE reported in this work.

\section*{Acknowledgment}

This work was supported by the ERC under Grant Agreement 101125691.

\appendices

\section{Proof of Theorem~\ref{thm:rate_bounds}}
\label{app:proof_rate_bounds}
 
We start by proving the expression for the effective communication noise $\sigma_{0,c}^2$ in \eqref{eq:sigma0_nocsit}, and then prove the desired upper and lower bounds through  comparison of  the quantized communication channel
\eqref{eq:quantized_comm_model} with an auxiliary Gaussian-noise channel of 
same noise covariance. 
 
\subsection{Effective Communication Noise $\sigma_{0,c}^2$}
\label{App:sigma2}
 

Under isotropic Gaussian signaling
$\mathbf{R}_{\bsf{X}}=P_0\mathbf{I}_{N_t}$, and since
$\mathbb{E}[\bsf{H}\bsf{H}^H]=N_t\mathbf{I}_{N_c}$:
\begin{equation}
    \mathbb{E}\!\left[\bsf{y}_{c,t}\bsf{y}_{c,t}^H\right]
    =
    P_0\,\mathbb{E}[\bsf{H}\bsf{H}^H]+\sigma_c^2\mathbf{I}_{N_c}
    =
    (N_tP_0+\sigma_c^2)\mathbf{I}_{N_c}.
    \label{eq:app_avg_rx_power}
\end{equation}
The semi-unitary constraint \eqref{eq:semi_unitary_combiner} implies
$\mathbf{R}_{\bsf{u}_c}
=\mathbf{A}_c\,\mathbb{E}[\bsf{y}_{c,t}\bsf{y}_{c,t}^H]\,\mathbf{A}_c^H
=(N_tP_0+\sigma_c^2)\mathbf{I}_{\tilde N_c}$, so that
$\max_i(\mathbf{R}_{\bsf{u}_c})_{i,i}=N_tP_0+\sigma_c^2$, irrespective of
$\mathbf{A}_c$ and of the realization $\mathbf{H}$. Substituting into
\eqref{eq:effective_noise_variance} recovers \eqref{eq:sigma0_nocsit}; in
particular, $\sigma_{0,c}^{2}$ is a constant that depends neither on the combiner
nor on the channel.
 \subsection{Lower bound}
\label{subsec:app_lower}
 
Since the input symbols are i.i.d.\ across $t$ and the channel is memoryless with
effective noise i.i.d.\ across $t$, the conditional mutual information
single-letterizes,
\begin{equation}
    \frac{1}{T}\,I(\bsf{X};\bsf{Z}_c\mid\bsf{H}=\mathbf{H})
    =
    I(\bsf{x}_1;\bsf{z}_{c,1}\mid\bsf{H}=\mathbf{H}).
    \label{eq:app_single_letter}
\end{equation}
We therefore drop the time index and write
$\bsf{z}_c=\mathbf{A}_c\mathbf{H}\bsf{x}+\bar{\bsf{w}}_c$ with
$\bsf{x}\sim\mathcal{CN}(\mathbf{0},P_0\mathbf{I}_{N_t})$. 
 
 
Notice next that by the worst additive-noise property of the Gaussian distribution \cite{diggavi2001worst}, we obtain a lower bound on the mutual information $I(\bsf{x};\bsf{z}_c\mid\bsf{H}=\mathbf{H})$ by replacing the effective noise $\bar{\bsf{w}}_c$ with a CSCG vector of the same covariance matrix. 
$\sigma_{0,c}^{2}\mathbf{I}_{\tilde N_c}$. 
For isotropic  Gaussian inputs
$\bsf{x}\sim\mathcal{CN}(\mathbf{0},P_0\mathbf{I}_{N_t})$, this yields the lower bound 
\begin{IEEEeqnarray}{rCl}
  I(\bsf{x};\bsf{z}_c\mid\bsf{H}=\mathbf{H}) &\geq &  \log\det\!\left(   \mathbf{I}_{\tilde N_c}
        +
        \frac{P_0}{\sigma_{0,c}^{2}}\,
        \mathbf{A}_c\mathbf{H}\mathbf{H}^H\mathbf{A}_c^H
    \right) \IEEEeqnarraynumspace \\
    & \triangleq & g(\mathbf{A}_c).
    \label{eq:app_g_def}
\end{IEEEeqnarray}

Let now $\mathbf{H}\mathbf{H}^H=\sum_{i=1}^{N_c}\lambda_i(\bsf{H}\bsf{H}^H)\mathbf{v}_i\mathbf{v}_i^H$
be its eigendecomposition with
$\lambda_1(\bsf{H}\bsf{H}^H)\ge\cdots\ge\lambda_{N_c}(\bsf{H}\bsf{H}^H)$. Since
$\mathbf{A}_c\mathbf{A}_c^H=\mathbf{I}_{\tilde N_c}$, the matrix
$\mathbf{A}_c\mathbf{H}\mathbf{H}^H\mathbf{A}_c^H$ is a compression of
$\mathbf{H}\mathbf{H}^H$, and the Poincar\'e separation theorem
\cite[Thm.~4.3.28]{horn2012matrix} yields
\begin{equation}
    \lambda_i\!\left(\mathbf{A}_c\mathbf{H}\mathbf{H}^H\mathbf{A}_c^H\right)
    \le
    \lambda_i(\bsf{H}\bsf{H}^H),
    \qquad
    i=1,\ldots,\tilde N_c,
    \label{eq:app_poincare}
\end{equation}
with equality for all $i$ simultaneously when the rows of $\mathbf{A}_c$ are the
dominant eigenvectors $\mathbf{v}_1,\ldots,\mathbf{v}_{\tilde N_c}$. As
$g(\mathbf{A}_c)$ 
is increasing in each eigenvalue, the maximizer is
\begin{equation}
    \mathbf{A}_c^{\circ}(\mathbf{H})
    =
    \begin{bmatrix}
        \mathbf{v}_1 & \cdots & \mathbf{v}_{\tilde{N}_c}
    \end{bmatrix}^H,
    \label{eq:optimal_comm_combiner}
\end{equation}
and
\begin{equation}
    \max_{\mathbf{A}_c}g(\mathbf{A}_c)
    =
    \sum_{i=1}^{\tilde N_c}
    \log\!\left(
        1+\frac{P_0\lambda_i(\bsf{H}\bsf{H}^H)}{\sigma_{0,c}^{2}}
    \right).
    \label{eq:app_g_max}
\end{equation}

  Taking expectation over $\bsf{H}$ yields the desired lower bound  in 
\eqref{eq:cg_nocsit}. 

Fix $\mathbf{H}$ and any semi-unitary $\mathbf{A}_c$. The effective noise
$\bar{\bsf{w}}_c$ has zero mean, covariance
$\sigma_{0,c}^{2}\mathbf{I}_{\tilde N_c}$, and is independent of the Gaussian input.
Since Gaussian noise is the worst case additive noise of given noise covariance matrix \cite{diggavi2001worst},  
replacing
$\bar{\bsf{w}}_c$ by a CSCG vector of the same covariance can only decrease the
mutual information,
\begin{equation}
    I(\bsf{x};\bsf{z}_c\mid\bsf{H}=\mathbf{H})
    \ge
    g(\mathbf{A}_c).
    \label{eq:app_lb_pointwise}
\end{equation}
Evaluating \eqref{eq:app_lb_pointwise} at $\mathbf{A}_c^{\circ}(\mathbf{H})$ and
using \eqref{eq:app_g_max},
\begin{equation}
    \max_{\mathbf{A}_c}I(\bsf{x};\bsf{z}_c\mid\bsf{H}=\mathbf{H})
    \ge
    \sum_{i=1}^{\tilde N_c}
    \log\!\left(
        1+\frac{P_0\lambda_i(\bsf{H}\bsf{H}^H)}{\sigma_{0,c}^{2}}
    \right).
    \label{eq:app_lb_perH}
\end{equation}
Taking  expectation over $\bsf{H}$  yields the desired lower bound $R_{\textnormal{G}}\ge\RGL$.
 
\subsection{Upper bound}
\label{subsec:app_upper}
 
Fix $\mathbf{H}$ and any semi-unitary $\mathbf{A}_c$. Since the effective noise is
independent of the input,
\begin{equation}
    I(\bsf{x};\bsf{z}_c\mid\bsf{H}=\mathbf{H})
    =
    h(\bsf{z}_c\mid\bsf{H}=\mathbf{H})
    -
    h(\bar{\bsf{w}}_c).
    \label{eq:app_ub_split}
\end{equation}
The output covariance is
$\mathbf{R}_{\bsf{z}_c\mid\mathbf{H}}
=P_0\mathbf{A}_c\mathbf{H}\mathbf{H}^H\mathbf{A}_c^H+\sigma_{0,c}^{2}\mathbf{I}_{\tilde N_c}$.
Among all complex random vectors with a given covariance, the
circularly-symmetric Gaussian maximizes differential entropy
\cite[Thm.~8.6.5]{cover2006elements}, so
\begin{IEEEeqnarray}{rCl}
    h(\bsf{z}_c\mid\bsf{H}=\mathbf{H})
    & \le &
    \log\det\!\left(\pi e\,\mathbf{R}_{\bsf{z}_c\mid\mathbf{H}}\right)
    \nonumber\\
    & = &
    \tilde N_c\log(\pi e\,\sigma_{0,c}^{2})+g(\mathbf{A}_c).
    \label{eq:app_output_entropy}
\end{IEEEeqnarray}
For the noise term, the semi-unitary combiner keeps the receiver noise white, so
$\bar{\bsf{w}}_c$ has i.i.d.\ entries $\bar{\mathsf{W}}_{c,i}$, each with
independent and identically distributed real and imaginary parts; therefore
\begin{equation}
    h(\bar{\bsf{w}}_c)
    =
    \tilde N_c\,h(\bar{\mathsf{W}}_{c})
    =
    2\tilde N_c\,h\!\left(\Re(\mathsf{\bar{W}}_{c})\right),
    \label{eq:app_noise_entropy}
\end{equation}
where $\Re(\mathsf{\bar{W}}_{c})$ is the sum of a $\mathcal{N}(0,\sigma_c^2/2)$
term and an independent $\mathcal{U}[-\Delta_c/2,\Delta_c/2]$ term, with the
density \eqref{eq:real_eff_noise_pdf}. Combining
\eqref{eq:app_ub_split}--\eqref{eq:app_noise_entropy}, only $g(\mathbf{A}_c)$
depends on the combiner; maximizing over $\mathbf{A}_c$ via \eqref{eq:app_g_max}
and taking  expectation over $\bsf{H}$ yields the desired upper bound in  
\eqref{eq:RGU}.
 
Above proof does not rely on the Gaussian signaling assumption, and thus $R_{\textnormal{G,U}}$ is also an upper bound on capacity. 

\section{Proof of Theorem~\ref{pro:pcs_sensing_bound}}
\label{app:proof_pcs_sensing_bound}

We first derive the LMMSE error
$\sigma_{\bsf{g}\mid\bsf{X}=\mathbf{X}}^2(\mathbf{A}_s)$ for a fixed realization
$\bsf{X}=\mathbf{X}$ and combiner $\mathbf{A}_s$ in closed form, and then take the
expectation over $\bsf{X}$ and minimize over $\mathbf{A}_s$.
 
\subsection{Conditional LMMSE for a fixed $\bsf{X}=\mathbf{X}$}
\label{subsec:app_conditional}
 
Taking the trace of the error covariance \eqref{eq:error_cov}, substituting
$\mathbf{R}_{\bsf{g}}=\mathbf{R}_{B}\otimes\mathbf{R}_{A}$ and
$\mathbf{R}_{\bar{\bsf{w}}_s}=\sigma_{0,s}^2\mathbf{I}_{\tilde N_sT}$ from
\eqref{eq:effective_noise_covariance}, and applying the mixed-product rule
$(\mathbf{A}\otimes\mathbf{B})(\mathbf{C}\otimes\mathbf{D})=\mathbf{A}\mathbf{C}\otimes\mathbf{B}\mathbf{D}$
so that every factor separates across the two Kronecker blocks, trace cyclicity
yields
\begin{IEEEeqnarray}{rCl}
    \lefteqn{\sigma_{\bsf{g}\mid\bsf{X}=\mathbf{X}}^2(\mathbf{A}_s)} \nonumber \\
    &=&
    \frac{1}{N_tN_s}
    \operatorname{Tr}(\mathbf{R}_{B}\otimes\mathbf{R}_{A})
    \notag\\
    &&
    -\frac{1}{N_tN_s}
    \operatorname{Tr}\!\biggl(
    \bigl(
        \mathbf{X}^{T}\mathbf{R}_{B}^{2}\mathbf{X}^{*}
        \otimes
        \mathbf{A}_s\mathbf{R}_{A}^{2}\mathbf{A}_s^H
    \bigr)
    \notag\\
    &&\hspace{1.9cm} \cdot
    \Bigl[
        \mathbf{X}^{T}\mathbf{R}_{B}\mathbf{X}^{*}
        \otimes
        \mathbf{A}_s\mathbf{R}_{A}\mathbf{A}_s^H
        +
        \sigma_{0,s}^{2}\mathbf{I}_{\tilde N_sT}
    \Bigr]^{-1}
    \biggr).
    \nonumber \\\label{eq:app_pcs_conditional_lmmse}
\end{IEEEeqnarray}
The first trace does not depend on $\mathbf{X}$ and we keep it unchanged. In the
following, we first simplify the effective sensing noise variance
$\sigma_{0,s}^2$, and then show that the second trace, abbreviated by
$f(\mathbf{X})$, reduces to a scalar sum.
 
From the stacked sensing model, the pre-quantization vector is
$\bsf{u}_s=(\bsf{X}^T\otimes\mathbf{A}_s)\bsf{g}+(\mathbf{I}_T\otimes\mathbf{A}_s)\bsf{w}_s$,
so that its covariance matrix is
\begin{equation}
    \mathbf{R}_{\bsf{u}_s}
    =
    \mathbb{E}_{\bsf{X}}\!\left[\bsf{X}^{T}\mathbf{R}_{B}\bsf{X}^{*}\right]
    \otimes
    \mathbf{A}_s\mathbf{R}_{A}\mathbf{A}_s^H
    +
    \sigma_s^2\mathbf{I}_{\tilde N_sT},
    \label{eq:app_Rus}
\end{equation}
and, recalling
$\sigma_{0,s}^{2}=\sigma_s^2+\beta\max_{i}(\mathbf{R}_{\bsf{u}_s})_{i,i}$ from
\eqref{eq:effective_noise_variance}, the effective sensing noise variance is
\begin{IEEEeqnarray}{rCl}
    \sigma_{0,s}^{2}
    & = &
    \beta
    \max_{i=1,\ldots,\tilde N_sT}
    \left(
        \mathbb{E}_{\bsf{X}}\!\left[\bsf{X}^{T}\mathbf{R}_{B}\bsf{X}^{*}\right]
        \otimes
        \mathbf{A}_s\mathbf{R}_{A}\mathbf{A}_s^H
    \right)_{i,i}
    \nonumber \\
    && +\, \sigma_s^2(1+\beta).
    \label{eq:app_sigma0_s_pre}
\end{IEEEeqnarray}
We continue to simplify this expression. For i.i.d.\ columns
$\bsf{x}_t\sim\mathcal{CN}(\mathbf 0,\mathbf{R}_{\bsf{X}})$, the $(t_1,t_2)$ entry
of $\mathbb{E}_{\bsf{X}}[\bsf{X}^{T}\mathbf{R}_{B}\bsf{X}^{*}]$ is zero for
$t_1\neq t_2$ and equals $\operatorname{Tr}(\mathbf{R}_{\bsf{X}}^{*}\mathbf{R}_{B})$
for $t_1=t_2$, so that
$\mathbb{E}_{\bsf{X}}[\bsf{X}^{T}\mathbf{R}_{B}\bsf{X}^{*}]
=\operatorname{Tr}(\mathbf{R}_{\bsf{X}}^{*}\mathbf{R}_{B})\mathbf{I}_{T}$.
Introducing $\bar{\mathbf{A}}_s\triangleq\mathbf{A}_s\mathbf{R}_A^{1/2}$, so that
$\mathbf{A}_s\mathbf{R}_{A}\mathbf{A}_s^H=\bar{\mathbf{A}}_s\bar{\mathbf{A}}_s^H$,
\eqref{eq:app_sigma0_s_pre} becomes
\begin{equation}
    \sigma_{0,s}^{2}
    =
    \beta\, \operatorname{Tr}(\mathbf{R}_{\bsf{X}}^{*}\mathbf{R}_{B})
    \max_{i}\left(\bar{\mathbf{A}}_s\bar{\mathbf{A}}_s^H\right)_{i,i}
    +\sigma_s^2(1+\beta).
    \label{eq:app_sigma0_s_diag}
\end{equation}
 
Both $f(\mathbf{X})$ and $\mathbf{R}_{\bsf{u}_s}$ depend on the combiner only
through $\bar{\mathbf{A}}_s$, and, as shown later in this appendix, the left singular vectors of $\bar{\mathbf{A}}_s$ affect $f(\mathbf{X})$ through $\max_{i}(\mathbf{R}_{\bsf{u}_s})_{i,i}$. We
therefore choose them so as to minimize $\max_{i}(\mathbf{R}_{\bsf{u}_s})_{i,i}$.
Since any Hermitian PSD matrix $\mathbf{M}$ satisfies
$\min_{\bar{\mathbf{U}}}\max_{i}(\bar{\mathbf{U}}\mathbf{M}\bar{\mathbf{U}}^H)_{i,i}
=\operatorname{Tr}(\mathbf{M})/\tilde N_s$
\cite[Cor.~2.4]{palomar2007majorization}, and since
$\operatorname{Tr}(\bar{\mathbf{A}}_s\bar{\mathbf{A}}_s^H)=\sum_{i}\sigma_i^2$,
where $\{\sigma_i\}_{i=1}^{\tilde N_s}$ are the singular values of
$\bar{\mathbf{A}}_s$, the effective noise variance becomes
\begin{equation}
    \sigma_{0,s}^2
    =
    \frac{\beta}{\tilde N_s}\operatorname{Tr}\!\left(\mathbf{R}_{\bsf{X}}^{*}\mathbf{R}_{B}\right)
    \sum_{i=1}^{\tilde N_s}\sigma_i^2
    +
    \sigma_s^2(1+\beta).
    \label{eq:app_sigma0_s_final}
\end{equation}
 
The feasible set of the squared singular values
$\{\sigma_i^2\}_{i=1}^{\tilde{N}_s}$ follows from
$\sum_{i}\sigma_i^2=\operatorname{Tr}(\bar{\mathbf{A}}_s\bar{\mathbf{A}}_s^H)
=\operatorname{Tr}(\mathbf{A}_s\mathbf{R}_{A}\mathbf{A}_s^H)$ and the semi-unitary
constraint $\mathbf{A}_s\mathbf{A}_s^H=\mathbf{I}_{\tilde N_s}$. Denoting by
$\lambda_{A,1},\ldots,\lambda_{A,N_s}$ the eigenvalues of $\mathbf{R}_{A}$ in
decreasing order, Ky Fan's eigenvalue inequalities
\cite[Cor.~4.3.39]{horn2012matrix} bound the trace
$\operatorname{Tr}(\mathbf{A}_s\mathbf{R}_{A}\mathbf{A}_s^H)$ between the sums of
the $\tilde N_s$ smallest and the $\tilde N_s$ largest eigenvalues of
$\mathbf{R}_{A}$,
\begin{equation}
    \sum_{i=N_s-\tilde N_s+1}^{N_s}\lambda_{A,i}
    \;\le\;
    \sum_{i=1}^{\tilde N_s}\sigma_i^2
    \;\le\;
    \sum_{i=1}^{\tilde N_s}\lambda_{A,i}.
    \label{eq:app_sigma_constraint}
\end{equation}
 
We next show that, once optimized over $\bar{\mathbf{A}}_s$, the term
$f(\mathbf{X})$ reduces to a scalar sum. To this end, for a given realization
$\mathbf{X}$, consider the singular value decompositions
\begin{equation}
    \mathbf{R}_{B}^{1/2}\mathbf{X}^{*}
    =
    \mathbf{U}'\boldsymbol{\Sigma}'\mathbf{V}'^H,
    \qquad
    \bar{\mathbf{A}}_s
    =
    \mathbf{U}_s\boldsymbol{\Sigma}_s\mathbf{V}_s^H,
    \label{eq:app_svd_decompositions}
\end{equation}
with $\boldsymbol{\Sigma}'\in\mathbb{C}^{N_t\times T}$ and
$\boldsymbol{\Sigma}_s\in\mathbb{C}^{\tilde N_s\times N_s}$ rectangular diagonal.
Substituting \eqref{eq:app_svd_decompositions} into $f(\mathbf{X})$, using
$\mathbf{X}^{T}\mathbf{R}_B^{1/2}=(\mathbf{R}_B^{1/2}\mathbf{X}^{*})^H$ and the
trace identity
$\operatorname{Tr}(\mathbf{M}_1\mathbf{M}_2)=\operatorname{Tr}(\mathbf{M}_2\mathbf{M}_1)$,
the common unitary factor $\mathbf{V}'\otimes\mathbf{U}_s$ cancels and we obtain
\begin{equation}
    f(\mathbf{X})
    =
    \operatorname{Tr}\!\Bigl[
        (\mathbf{U}'^H\mathbf{R}_{B}\mathbf{U}'\otimes\mathbf{V}_s^H\mathbf{R}_{A}\mathbf{V}_s)\,
        \mathbf{D}
    \Bigr],
    \label{eq:app_f_with_D}
\end{equation}
with the diagonal matrix
$\mathbf{D}\triangleq(\boldsymbol{\Sigma}'\otimes\boldsymbol{\Sigma}_s^H)
(\boldsymbol{\Sigma}'^H\boldsymbol{\Sigma}'\otimes\boldsymbol{\Sigma}_s\boldsymbol{\Sigma}_s^H+\sigma_{0,s}^2\mathbf{I})^{-1}
(\boldsymbol{\Sigma}'^H\otimes\boldsymbol{\Sigma}_s)$.
In particular, $f(\mathbf{X})$ does not depend on $\mathbf{U}_s$, which is the
left-unitary freedom invoked above.
 
Let $\lambda'_n\triangleq(\sigma'_n)^2$, $n=1,\dots,N_t$, denote the nonzero
eigenvalues of
$\mathbf{R}_{B}^{1/2}\bsf{X}^{*}\bsf{X}^{T}(\mathbf{R}_{B}^{1/2})^H$, i.e., the
nonzero diagonal entries of $\boldsymbol{\Sigma}'\boldsymbol{\Sigma}'^H$, and
recall that $\sigma_i$, $i=1,\ldots,\tilde N_s$, are the diagonal entries of
$\boldsymbol{\Sigma}_s$. The $\bigl((n-1)N_s+i\bigr)$-th diagonal entry of
$\mathbf{D}$ is then
\begin{equation}
    D_{(n-1)N_s+i}
    =
    \frac{\lambda'_n\,\sigma_i^2}{\lambda'_n\,\sigma_i^2+\sigma_{0,s}^2}.
    \label{eq:app_D_entries}
\end{equation}
By the diagonality of $\mathbf{D}$, only the diagonal entries of
$\mathbf{U}'^{H}\mathbf{R}_{B}\mathbf{U}'$ and $\mathbf{V}_s^H\mathbf{R}_{A}\mathbf{V}_s$
contribute to the trace in \eqref{eq:app_f_with_D}, so that
\begin{equation}
    f(\mathbf{X})
    =
    \sum_{n=1}^{N_t}\sum_{i=1}^{\tilde N_s}
    d_{B,n}\,d_{A,i}\,
    \frac{\lambda'_n\,\sigma_i^2}{\lambda'_n\,\sigma_i^2+\sigma_{0,s}^2},
    \label{eq:app_f_scalar}
\end{equation}
where $d_{B,n}\triangleq(\mathbf{U}'^{H}\mathbf{R}_{B}\mathbf{U}')_{n,n}$ and
$d_{A,i}\triangleq(\mathbf{V}_s^H\mathbf{R}_{A}\mathbf{V}_s)_{i,i}$.
 
The vector $(d_{A,1},\dots,d_{A,N_s})$ is majorized by the vector of eigenvalues
$(\lambda_{A,1},\ldots,\lambda_{A,N_s})$ of $\mathbf{R}_{A}$, with equality
attained by choosing $\mathbf{V}_s$ to align with the eigenvectors of
$\mathbf{R}_{A}$. Since $f(\mathbf{X})$ is a linear functional of the $d_{A,i}$
with nonnegative coefficients, majorization theory
\cite[Cor.~2.1]{palomar2007majorization} implies that it is maximized by this
aligned choice, for which $d_{A,i}=\lambda_{A,i}$, $i=1,\ldots,\tilde N_s$.
Equation \eqref{eq:app_f_scalar} then becomes
\begin{equation}
    f(\mathbf{X})
    =
    \sum_{n=1}^{N_t}\sum_{i=1}^{\tilde N_s}
    d_{B,n}\,\lambda_{A,i}\,
    \frac{\lambda'_n\,\sigma_i^2}{\lambda'_n\,\sigma_i^2+\sigma_{0,s}^2}.
    \label{eq:app_f_scalar_lambda}
\end{equation}
Combining the above, an optimal sensing combiner is
$\mathbf{A}_s^{\circ}=\mathbf{U}_s\boldsymbol{\Sigma}_s\mathbf{V}_s^H\mathbf{R}_A^{-1/2}$.
 
\subsection{Expectation over $\bsf{X}$}
\label{subsec:app_expectation}
 
In \eqref{eq:app_f_scalar_lambda}, the only quantities that depend on the random
waveform $\bsf{X}$ are the eigenvalues $\{\lambda'_n\}$ of
$\bsf{W}=\mathbf{R}_B^{1/2}\bsf{X}^{*}\bsf{X}^{T}(\mathbf{R}_B^{1/2})^H$ and the
diagonal entries $\{d_{B,n}\}$ of $\mathbf{U}'^{H}\mathbf{R}_{B}\mathbf{U}'$, where
$\mathbf{U}'$ is the eigenvector matrix of $\bsf{W}$. Since $\bsf{W}$ is a
Wishart-type matrix, its eigenvectors are isotropically distributed and
statistically independent of its eigenvalues $\{\lambda'_n\}$
\cite[Thm.~2.2]{couillet2011random}. As $d_{B,n}$ is a function of $\mathbf{U}'$
only, this independence lets the expectation of each summand in
\eqref{eq:app_f_scalar_lambda} factor as
\begin{equation}
    \mathbb{E}_{\bsf{X}}\!\left[
        d_{B,n}\,\frac{\lambda'_n\,\sigma_i^2}{\lambda'_n\,\sigma_i^2+\sigma_{0,s}^2}
    \right]
    =
    \mathbb{E}_{\bsf{X}}[d_{B,n}]\;
    \mathbb{E}_{\bsf{X}}\!\left[
        \frac{\lambda'_n\,\sigma_i^2}{\lambda'_n\,\sigma_i^2+\sigma_{0,s}^2}
    \right].
    \label{eq:app_factor}
\end{equation}
By the isotropy of the eigenvectors, each column $\mathbf{u}'_n$ of $\mathbf{U}'$
is uniformly distributed on the unit sphere, hence
$\mathbb{E}_{\bsf{X}}[\mathbf{u}'_n\mathbf{u}_n'^{H}]=\mathbf{I}_{N_t}/N_t$, and by
the cyclicity of the trace,
\begin{IEEEeqnarray}{rCl}
    \mathbb{E}_{\bsf{X}}[d_{B,n}]
    &=&
    \mathbb{E}_{\bsf{X}}\!\bigl[(\mathbf{U}'^{H}\mathbf{R}_{B}\mathbf{U}')_{n,n}\bigr]
    =
    \mathbb{E}_{\bsf{X}}\!\bigl[\mathbf{u}_n'^{H}\mathbf{R}_{B}\mathbf{u}_n'\bigr]
    \nonumber\\
    &=&
    \operatorname{Tr}\!\bigl(\mathbf{R}_{B}\,\mathbb{E}_{\bsf{X}}[\mathbf{u}'_n\mathbf{u}_n'^{H}]\bigr)
    =
    \frac{\operatorname{Tr}(\mathbf{R}_{B})}{N_t}.
    \IEEEeqnarraynumspace
    \label{eq:app_EdB}
\end{IEEEeqnarray}
Substituting \eqref{eq:app_factor} and \eqref{eq:app_EdB} into the expectation of
\eqref{eq:app_f_scalar_lambda} gives
\begin{equation}
    \mathbb{E}_{\bsf{X}}[f(\bsf{X})]
    =
    \sum_{n=1}^{N_t}\sum_{i=1}^{\tilde N_s}
    \frac{\operatorname{Tr}(\mathbf{R}_{B})}{N_t}\,\lambda_{A,i}\,
    \mathbb{E}_{\bsf{X}}\!\left[
    \frac{\lambda'_n\,\sigma_i^2}{\lambda'_n\,\sigma_i^2+\sigma_{0,s}^2}
    \right].
    \label{eq:app_Ef}
\end{equation}
 
Finally, the conditional distortion \eqref{eq:app_pcs_conditional_lmmse} equals
$\frac{1}{N_tN_s}\operatorname{Tr}(\mathbf{R}_B\otimes\mathbf{R}_A)
-\frac{1}{N_tN_s}f(\mathbf{X})$, so minimizing the expected distortion over the
sensing combiner amounts to maximizing $\mathbb{E}_{\bsf{X}}[f(\bsf{X})]$ over the
squared singular values $\{\sigma_i^2\}$ subject to
\eqref{eq:app_sigma_constraint}. Substituting \eqref{eq:app_Ef} and
\eqref{eq:app_sigma0_s_final} into this maximization yields
\eqref{eq:epsilon_G_general}, which completes the proof.

\bibliographystyle{IEEEtran}
\bibliography{references_abbrev_compact}
\vspace{12pt}

\end{document}